\documentclass[a4paper,UKenglish]{lipics}

\pdfoutput=1

\def\disc{0} 
\def\submit{1}

\usepackage{amsthm}
\usepackage{setspace} 
\usepackage{xspace}
\usepackage{paralist}
\usepackage{enumerate}
\usepackage{comment}
\usepackage{url}
\usepackage{amssymb}
\usepackage{amsmath}
\usepackage{paralist}
\usepackage[active]{srcltx}
\usepackage{algorithm}
\usepackage{algpseudocode} 
\usepackage{bbm}
\usepackage{graphicx}
\graphicspath{{./Figs/}}
\usepackage{color} 
\usepackage{latexsym}

\newenvironment{proof sketch}[1]{\noindent {\emph{Proof sketch of #1:}}}{\hfill \qed}
\newenvironment{sketch}{\noindent {\emph{Proof sketch:}}}{\hfill \qed}
\newcommand{\eqdf}{\triangleq}

\newcommand{\bmax}{b_{\max}}

\newcommand{\af}{\Phi}
\newcommand{\pmax}{p_{\max}}

\newcommand{\ppn}{\mbox{\rm pn}}
\newcommand{\ppr}{\mbox{\rm pr}}
\newcommand{\fold}{\mbox{\emph{fold}}}
\newcommand{\expand}{\mbox{\emph{expand}}}

\newcommand{\sou}{s}
\newcommand{\de}{t}

\newcommand{\TSJ}{\alpha_j}

\newcommand{\NN}{{\mathbb{N}}}

\newcommand{\event}{\text{$\sigma$}}

\newcommand{\load}{\text{\textit{load}}}

\newcommand{\routealg}{\text{\textsc{Route}}}
\newcommand{\unroutealg}{\text{\textsc{UnRoute}}}

\newcommand{\alg}{\textsc{alg}\xspace}
\newcommand{\AAP}{\textsc{allocate}}
\newcommand{\raper}{\textsc{ra-persist}}
\newcommand{\benefit}{\textit{benefit}}

\newcommand{\accstdby}{\textsf{Acc/Stdby}}

\newcommand{\optp}{\textsc{opt}}

\newcommand{\optpack}{\textsc{opt}}

\newcommand{\pr}{{\sc pr}}

\newcommand{\Set}[1]{\left\{#1\right\}}

\usepackage{color}
\usepackage{soul}

\renewcommand{\paragraph}[1]{\par\noindent\textbf{#1}}

\title{Competitive Path Computation and Function Placement in SDNs%
\footnote{This work was supported in part by the Neptune Consortium, Israel.}
}

\author[1]{Guy Even}
\author[2]{Moti Medina}
\author[3]{Boaz Patt-Shamir}
\affil[1]{School of Electrical Engineering\\ Tel Aviv University, Tel
  Aviv 6997801, Israel\\{\tt guy@eng.tau.ac.il}}
\affil[2]{MPI for Informatics\\66123 Saarbr\"ucken, Germany\\
{\tt mmedina@mpi-inf.mpg.de}}
\affil[3]{School of Electrical Engineering\\ Tel Aviv University, Tel
  Aviv 6997801, Israel\\{\tt boaz@tau.ac.il}}
\begin{document}
\Copyright{G.~Even, M.~Medina and B.~Patt-Shamir}

\maketitle

\begin{abstract}
We consider a task of serving requests that arrive in an online
fashion in Software-Defined Networks (SDNs) with network function
virtualization (NFV).  Each request specifies an abstract routing
and processing ``plan'' for a flow (e.g., the flow source is at node
$s$ and it needs to reach destination node $t$ after undergoing a
few processing stages, such as firewall or encryption). Each
processing function can be performed by a specified subset of
servers in the system. The algorithm needs to either reject the
request or admit it and return detailed routing (a.k.a.\ ``path
computation'') and processing assignment (``function placement'').
Each request also specifies the communication bandwidth and the
processing load it requires. Components in the system (links and
processors) have bounded capacity; a feasible solution may not
violate the capacity constraints.  Requests have benefits and the
goal is to maximize the total benefit of accepted requests.

In this paper we first formalize the problem, and propose a new
service model that allows us to cope with requests with \emph{unknown
duration}.  The new service model augments the traditional accept/reject
schemes with a new possible response of ``stand by.''  Our main result
is an online algorithm for path computation and function placement
that guarantees, \emph{in each time step}, throughput of at least
$\Omega\left(\frac{\text{OPT}^*}{\log n}\right)$, where $n$ is the
system size and OPT$^*$ is an upper bound on the
maximal possible throughput.  The guarantee holds assuming that
requests ask
 for at most an
$O\left(1/{\log n}\right)$-fraction of the capacity of any component
in the system.
Furthermore, the guarantee holds even though our
algorithm serves requests in an all-or-nothing fashion using a single
path and never preempts accepted flows, while OPT$^*$ may serve
fractional requests,
may split the allocation over multiple paths, and may arbitrarily
preempt and resume service of requests.
\end{abstract}

\section{Introduction}
Conventional wisdom has it that in networking, models are reinvented
every twenty years or so. A deeper look into the evolution of networks
shows that there is always a tension between ease of computation,
which favors collecting all data and performing processing centrally,
and ease of communication, which favors distributing the computation
over nodes along communication paths. It seems that recently the
pendulum has moved toward the centralized computation once again, with
the emergence of software-defined networks (SDNs), in which the underlying abstraction is of a centrally managed
network.

Among the key components of SDNs are \emph{path computation} and \emph{function
  placement}~\cite{SDNsurvey}, in which potentially complex requests need to be
routed over the network. Each request specifies a ``processing plan'' that includes a
source-destination pair as well as a description of a few processing stages that the stream
needs to go through.
The task
is to find a route in the network from the source to the
destination that includes the requested processing.
 The main difficulty, of course, is the bounded
processing capacity of servers and links, so not all
requests can be
served.

\paragraph{Our Contributions.} Path computation is often solved after
function placement.  Even if both tasks are solved (approximately)
optimally, the quality of the composed solution may not be good. In
this paper we solve these tasks together by a competitive on-line
algorithm. Our contribution is both conceptual and technical.
From the conceptual viewpoint, we introduce a new service model that facilitates competitive
non-preemptive algorithms
for requests whose duration is unknown upon arrival.
In the new
service model, a request which
is not admitted is placed in a ``standby'' queue until
there is room to accept it (due to other requests leaving
the system).
Once a request is accepted, it is guaranteed to receive service until
it
ends (i.e., until the user issues a ``leave'' signal).

Our algorithmic contribution consists of a 
deterministic algorithm that receives requests in an on-line fashion,
and
determines when each request starts receiving service (if at all), and how is this
service provided (i.e., how to route the request and where to process it).
 Each request has a benefit per time unit it receives
service, and the  algorithm is guaranteed to obtain $\Omega(1/\log
(nk))$ of the best possible benefit, where $n$ is the system size and
$k$ is the maximum number of
processing stages of a request.\footnote{Typically, $k$ is constant because the number
  of processing stages does not grow as a function of the size $n$ of
  the network.}
More precisely, in every time step $t$, the benefit collected by the
algorithm is at
least an $\Omega(1/\log (nk))$-fraction of the largest possible total
benefit that can be obtained at time $t$, i.e., from from all requests
that are active at time $t$ while respecting the capacity constraints.
The competitive ratio of the algorithm holds
under the conditions that no processing stage of a request requires more than an
$O\big(1/(k\log (nk))\big)$ fraction of the capacity of any component
(node or link) in the system, and assuming that the ratio
of the highest-to-lowest benefits of requests is bounded by
a polynomial in $n$. (We provide precise statements below.)
We also prove a lower bound on the competitive ratio of $\Omega(\log
n)$  for every online algorithm in our new model. Hence, for $k\in
n^{O(1)}$,  our algorithm
is asymptotically optimal. 

\subsection{Previous Work}
\emph{SDN Abstractions via High Level SDN Programming
Languages.}
Merlin~\cite{soule2014merlin,soule2013managing} is an SDN
language for provisioning network resources. In Merlin,
requests are specified as a regular expression with
additional annotation. The system works in an off-line
fashion: given a set of requests and the system
description, an integer linear program (ILP) is generated.
Then an external ILP solver is used to decide which
requests are accepted and how are they routed. Strictly
speaking, due to the use of an ILP solver, this solution is
not polynomial time. For more information on SDN languages
(and SDN in general) we refer the reader
to~\cite{SDNsurvey}.

\emph{Online Routing Algorithms.}
Our work leverages the seminal algorithm of Awerbuch
et.\ al~\cite{AAP}, which is an on-line algorithm for routing 
requests with given benefits and \emph{known durations}.
The algorithm of \cite{AAP} decides
whether to admit or
reject each request when it arrives; the algorithm also computes
routes for the admitted requests. The goal of the algorithm in
\cite{AAP}
is to
maximize the 
sum of benefits of accepted
requests.
 The benefit-maximization algorithm of \cite{AAP} resembles the
 cost-minimization algorithm
presented in \cite{aspnes1997line} for a different model: in \cite{aspnes1997line}, all
requests must be
accepted, and the goal is to minimize the
maximal ratio, over all links, between the total link load and link
capacity (the load of a link is the total
bandwidth of all requests routed through it).

Buchbinder and Naor~\cite{BN06,BN09} analyze the algorithm of ~\cite{AAP} using the
primal-dual method. This allows them to bound  the benefit of
the computed solution as a function of the benefit of an optimal fractional solution
(see also~\cite{kleinberg1996approximation}).

As mentioned, the above algorithms assume that each request specifies
the duration of the service it needs when it arrives.
The only on-line algorithm for unknown durations we know of in this
context is
for the problem of minimizing the
maximal load~\cite{awerbuch2001competitive}. The algorithm
in~\cite{awerbuch2001competitive} is $O(\log
n)$-competitive, but it requires rerouting of admitted requests
(each request may be rerouted
$O(\log n)$ times). Our algorithm, which is for benefit maximization,
deals with unknown durations without rerouting by allowing the
``standby'' mode.

\subsection{Advocacy of the Service Model}
In the classical non-preemptive model with guaranteed bandwidth,
requests must specify in advance what is the exact duration of the
connection (which may be infinite), and the system must give an immediate
response, which may be either ``reject'' or ``admit.'' While
immediate responses are preferable in general, the requirement that
duration  is specified in advance is unrealistic in many cases
(say, because the length of the connection may depend on yet-unavailable
inputs). However, requests with unknown durations seem to thwart
the possibility for a competitive algorithm due to the following
reasoning. Consider any system, and suppose that there are
infinitely many requests available at time $0$, all with unit benefit
per time step. Clearly there exists a request, say $r^*$, that is
rejected due to the finite capacity of the system. Now, the following
adversarial scenario may unfold: all admitted requests
leave the system at time $1$, and request $r^*$ persists
forever. Clearly, this means that no deterministic algorithm
can guarantee a non-trivial competitive ratio in the worst case.

We therefore argue that if unknown durations are to be tolerated, then
the requirement for an immediate reject/admit response must be
relaxed. One relaxation is to allow preemption, but with preemption
the connection is never certain until it terminates. Our service
model suggests to commit upon
accept, but not to
commit to rejection. This type of service is quite common in many daily
activities
(e.g., waiting in line for a restaurant seat), and is actually
implicitly present in some
admit/reject situations: in many cases, if a
request is rejected, the user will try to re-submit it. Moreover, from
a more philosophical point of view, the ``standby'' service model seems
fair for unknown durations: on
one hand, a request
does not commit ahead of time to when it will leave the system, and on
the other
hand, the algorithm does not commit ahead of time to when the request
will \emph{enter} the
system.

\ifnum\submit=0
\paragraph{Paper Organization.}
The remainder of the paper is organized as follows. In
Sec.~\ref{sec-model} we formalize the problem we solve.
In Sec.~\ref{sec-wrapper} we show how to implement a key
component in our algorithm. In
Sec.~\ref{sec-alg} we describe the main algorithm and
prove our main result.
\fi


\section{ 
	Request Model and Service Model}
\label{sec-model}
In this section we formalize the problem of path computation and
function mapping in SDNs. The main new concept in the way the input is
specified is called \emph{\pr-graphs}. The nodes
of a \pr-graph represent servers and the edges represent
communication paths, so that a \pr-graph is an abstract representation
of a request.\footnote{
Our \pr-graphs are similar to Merlin's regular
expressions~\cite{soule2014merlin}, but are more expressive and,
in our humble opinion, are more natural to design.
}
The main novelty of our output model is in allowing the system to put
arriving requests in a ``standby'' mode instead of immediately
rejecting them.
Details are provided in the remainder of this section.

\subsection{The Physical Network}
The network is a fixed network of servers and communication links. The network is
represented by a graph $N = (V,E)$, where $V$ is the set of \emph{nodes} and $E$ is
the set of \emph{edges}.
 Nodes and edges have \emph{capacities}. The capacity of
an edge $e$ is denoted by $c_e$, and the capacity of a node $v\in V$
is denoted by $c_v$.  All capacities are positive integers.
We note that the network is static and undirected (namely each edge
represents a bidirectional communication link), but may contain
parallel edges.
\subsection{Request Model and the Concept of \pr-Graphs}
\label{sec:problem}\label{sec-req-model}
Each request is a tuple $r_j = (G_j,d_j,b_j,U_j)$ with the following
interpretation.
\begin{compactitem}
\item $G_j=(X_j,Y_j)$ is a directed graph called the
    \pr-graph, where $X_j$ is the set of \pr-vertices, and $Y_j$ is the set of \pr-edges. We
  elaborate on the \pr-graph below.
\item $d_j: X_j\cup Y_j \rightarrow \NN$ is the
    \emph{demand}
  of the request from each \pr-graph component (i.e., bandwidth
  for links, processing for nodes). 
\item $b_j\in\NN$ is the \emph{benefit} paid by the
    request for each time step it is served.
\item $U_j:X_j\cup Y_j\to 2^V\cup 2^E$ maps each node in
    the \pr-graph
  to a set of nodes of $N$, and each edge in the \pr-graph is
  mapped to a set of edges of $N$. We elaborate below.
\end{compactitem}

\paragraph{The Processing and Routing Graph (\pr-graph).}
We refer to edges and vertices
in $G_j$ as \pr-edges and \pr-vertices, respectively. There are three types of
vertices in the \pr-graph $G_j$:
\begin{compactitem}
\item A single \emph{source} vertex $s_j\in X_j$ (i.e., vertex with
  in-degree zero) that represents the location 
  from
  which the packets arrive.
\item A single \emph{sink} vertex $t_j\in X_j$ (i.e., vertex with
  out-degree zero) that represents the location 
  to
  which the packets are destined.
\item \emph{Action vertices}, which represent
  transformations to be applied to the
    flow (such as encryption/decryption,  deep packet
    inspection, trans-coding etc.)
\end{compactitem}

\paragraph{Realization of \pr-paths and the $U$ function.}
The semantics of a \pr-graph is that the request can be served by any
source-sink path in the \pr-graph. However, these paths are abstract.
To interpret them
in the network, we map \pr-nodes to physical network nodes and
\pr-edges to physical network paths. To facilitate this mapping, each
request $r_j$ also includes the $U_j$ function, which, intuitively,
says which physical nodes (in $V$) can implement each \pr-node, and
which physical links (in $E$) can implement each \pr-edge.
Formally, we define the following concepts.

\begin{definition}
[valid realization of \pr-edge]
A simple path $p=(v_0,\ldots, v_k)$ in the network $N$
is a \emph{valid realization} of a \pr-edge $e$ if for all $0<i\le k$
we have that $(v_{i-1},v_i)\in U_j(e)$.
\end{definition}
Note that the empty path in $N$ is a valid realization of any
\pr-edge.

\begin{definition}[valid realizations of \pr-path]
	\label{def-pr-path}
  A path $p=(v_0,\ldots, v_k)$ in $N$ is a \emph{valid realization} of a path
  $\tilde{p}=(x_0,\ldots,x_{\ell})$ in $G_j$ under \emph{segmentation}
  $f:\Set{0,\ldots,\ell}\to\Set{0,\ldots,k}$ if
  \begin{compactitem}
  \item for all $0\le i\le\ell$, $v_{f(i)}\in U_j(x_i)$, and
  \item for all $0< i\le\ell$, the sub-path $(v_{f({i-1})},\ldots, v_{f({i})})$ of $p$ is a valid realization of $(x_{j-1},x_{j})$.
  \end{compactitem}
\end{definition}

The interpretation of mapping a \pr-node $x$ to a network node $v$ is
that the service represented by $x$ is implemented by $v$. $U_j(x)$ in
this case represents all physical nodes in which that service can be
performed.  Given a \pr-edge $e$, $U_j(e)$  is the set of links that may
be used to realize $e$. By default, $U_j(e)=E$, but
$U_j(e)$ allows the request designer to specify a set of edges to be
avoided due to any
consideration (e.g., security). Regarding processing,
consider the segmentation of the path in $N$ induced by
a valid realization.  The
endpoint of each
subpath is a network node in which the corresponding action takes
place. Moreover, the same network node may be
used for serving multiple actions for the same request.

\medskip \noindent
We are now ready to define the set of valid routings and processing for request
$r_j$.
\begin{definition}[valid realizations of request]
  A path $p$ in $N$ is a valid realization of a request $r_j$ if there exists a
  simple path $\tilde{p}$ in the \pr-graph $G_j$ from $s_j$ to $t_j$
  such that $p$ is a realization of $\tilde{p}$.
\end{definition}

\paragraph{Examples.}
Let us illustrate the utility of \pr-graphs with a few examples.

\emph{Simple Routing.}
 A request $r_j$ to route a connection from node $v$ to node $v'$ is
 modeled by a single-edge \pr-graph $s\stackrel{e}{\rightarrow} t$ with mappings
 $U_j(s)=\Set{v}$,
 $U_j(t)=\Set{v'}$, and $U_j(e)=E$. The demand from $e$ is the requested
 connection
 bandwidth. 

\emph{Serial Processing.}
A stream that needs to pass  $k$ transformations
  $a_1,\ldots,a_k$ in series is modeled by a path of $k+1$
  edges $s_j \rightarrow a_1 \rightarrow \cdots \rightarrow a_k
  \rightarrow t_j$, where $U_j(a_i)$ is the set of network nodes that
  can perform transformation $a_i$, for $i=1,\ldots,k$. Note that we can model bandwidth
  changes (e.g., if one of the transformations is compression) by
  setting different demands to different \pr-edges.

\emph{Regular Expressions.} Given any regular expression of processing
we can construct a \pr-graph by constructing the NFA
corresponding to the given expression \cite{HU}.

We note that our request model is more expressive than the
regular-expression model proposed by
Merlin~\cite{soule2014merlin}. For example, we can
model changing loads.%
\footnote{In Merlin, the input may also contain a ``policing''
  function of capping the maximal bandwidth of a connection. We focus
  on resource allocation only. Policing may be enforced by an
  orthogonal entity.  }

\smallskip

\paragraph{Capacity constraints and feasible realizations.}
Let $\tilde p=(s_j\stackrel{e_1}{\rightarrow} a_1,\ldots,a_k
\stackrel{e_{k+1}}{\rightarrow} t_j)$ denote a path in the \pr-graph
$G_j$.  Let $p=p_1\circ \cdots \circ p_{k+1}$ denote a valid
realization of $\tilde{p}$, where $p_i$ is a valid realization of
$e_i$. Let $v_i$ denote the endpoint of subpath $p_i$ for $1\leq
i\leq k$ ($v_i$ is where action $a_i$ takes place).  The load incurred
by serving request $r_j$ with demand $d_j$ by $p$ on each node and
edge in $p$ is
defined as follows (the load incurred on edges and nodes not in $p$ is
zero):
\begin{align*}
\load(v, p)&\triangleq \sum_{i: v=v_i}\frac{d_j(a_i)}{c_v}
\text{ for all }v\in\{v_1,\ldots,v_k\}\\ 
\load(e,p)&\triangleq\sum_{i: e\in p_i}\frac{ d_j(e_i)}{c_e}
\text{ for all } e\in p~.
\end{align*}
Informally, $\load(v, p)$ is the relative capacity of $v$ consumed by
$p$, and similarly $\load(v, e)$.

\begin{definition}[capacity constraints]
  Given a sequence of requests $\{r_j\}_{j\in I}$, a sequence of
  realizations $\{p^j\}_{j\in I}$
  satisfies the capacity constraints if
\begin{align*}
	\forall v\in V:&~ \sum_{j\in I} \load(v,p^j)\leq 1~\\
	\forall e\in E:&~\sum_{j\in I} \load(e,p^j)\leq 1~.
\end{align*}
\end{definition}
Given loads for nodes and edges, we say that a path $p$
from $\sou_j$ to $\de_j$ is a \emph{feasible realization} of request
$r_j$ if $p$ is a valid realization of an $s$-$t$ path of
$r$, and if $p$ satisfies the capacity constraints.

\subsection{The \accstdby\ Service Model}
\label{sec-serv-model}

We now describe the service model,
i.e., the user-system interface and guarantees.

\paragraph{Input.}
The input to the algorithm is the fixed network $N=(V,E)$ and a
sequence of events $\sigma=\{\event_t\}_{t \in \NN}$ which arrive one
at a time.  An event is either an arrival of a new request, or the
departure of a request that arrived earlier.  The attributes of an
arriving request $r_j$ are
as described in Sec.~\ref{sec:problem}, along with
 the  \emph{arrival time}
 of the request $\TSJ\in \NN$. We use $s_j$ and $t_j$ to denote the
source and the destination of the $j^{th}$ request.  A
departure event specifies  which request is departing and
the current time.

\paragraph{Output.}
The algorithm must generate a response to each arrival event, and may
generate any number of responses after a departure event. There are two
types of responses.
\par$\bullet$ \emph{Accept}: A request that has already arrived is accepted to the
system; the response also includes a feasible realization of the request. The request
will be served continuously from the time it is accepted until its departure event
(i.e., no preemption).  An ``accept'' response may follow any event; moreover,
multiple accepts (of multiple requests) are possible after a single event (typically
after a departure).

\par$\bullet$ \emph{Standby:} In this case an arriving request is not accepted
immediately, but may be accepted later.  When a request arrives, the system must
respond immediately by either accept or standby.

\paragraph{Performance Measure.} We evaluate algorithms by their
\emph{competitive ratio} \cite{SleatorT-85}.
Formally, given an algorithm $\alg$ and a finite input sequence
$\sigma$, let $\alg(\sigma)$ denote the total benefit $\alg$ receives
on input $\sigma$, where the system receives the benefit $b_j$ of
request $r_j$ for each time unit in which $r_j$ is served.
The {competitive
ratio} of an online algorithm $\alg$ for $\sigma$ is
$
\rho(\alg(\sigma)) \eqdf {\alg(\sigma)}/{\optp(\sigma)}
$,
where $\optp(\sigma)$ denotes the maximal possible benefit from
$\sigma$ by any allocation that respects the capacity constraints. The
competitive ratio of $\alg$ is 
\(
\rho(\alg) \eqdf \inf_{\sigma} \rho(\alg(\sigma)) \:.
\)


 \section{Computation of Light Valid Realizations}
\label{sec-wrapper}
\label{sec-reduction}

The algorithm presented
in Section~\ref{sec-alg} uses an ``oracle'' (subroutine) that finds a
feasible
realization of requests.
In this section we explain how to implement this oracle.

\subsection{Construction of Product Network and Product Request}

\paragraph{Input.}
 We are
given a weighted (physical) network $N=(V,E,w)$ with weights $w:V\cup E
\rightarrow \mathbb{R}^{\ge0}$ over nodes and edges, and a request
$r_j=(G_j,d_j,b_j,U_j)$, where
$G_j=(X_j,Y_j)$ is the \pr-graph with \pr-nodes $X_j$ and
\pr-edges $Y_j$ (cf.~Section~\ref{sec-req-model}). We are also given, for every \pr-node $x$
and \pr-edge $e$, the set of allowed nodes $U_j(x)\subseteq
V$ and edges $U_j(e)\subseteq E$, respectively.

\paragraph{Output: The product network.}
We construct the \emph{product network}, denoted $\ppn(N,r_j)$, which
is a weighted directed graph, with weights over nodes only.
The nodes of $\ppn(N,r_j)$, denoted $V'$, are $V'=V\times
Y_j$. The
edges of  $\ppn(N,r_j)$, denoted $E'=E_1\cup E_2$, are of two
categories, $E_1$ and $E_2$, defined as follows (we use $w$ to denote the
weight function in the product network too).
\begin{compactitem}
\item $\displaystyle
  E_1=\Set{\big((v,y),(v',y)\big) \mid y\in Y_j,\,
    (v,v')\in U_j(y)}$ (\emph{routing edges}).
\\
The weight of a routing edge is defined by
$w\big((v,y),(v',y)\big)\triangleq w{(v,v')}$, i.e., the weight of
the corresponding edge in $N$.

\item
$E_2=\big\{\big((v,y),(v,y')\big)
  \mid y,y'\in Y_j\text{ s.t. }y,y'\text{ share a node }x\text{ and }
  v\in  U_j(x)\big\}$ (\emph{processing edges}).\\
The weight of a processing edge is defined by
$w\big((v,y),(v,y')\big)=w(v)$, i.e., the weight of the corresponding
node in $N$.
\end{compactitem}

\paragraph{Output: The product request.}
The \emph{product request} $\ppr(N,r_j)$ is a pair of sets
$(S_j,T_j)$, called the \emph{source} and \emph{sink} sets, respectively.
The source set $S_j$ is
the set of all $(v,e)$ pairs such that  $v\in U(s_j)$ (i.e., $v$ is a
physical node that can be a source of the request $r_j$) and {$e$
    is incident to $v$ in $G_j$} (i.e., $e$ is a \pr-edge that can be
  the first edge in a source-sink path in the \pr-graph).
Similarly, the sink set $T_j$ is defined by
$T_j\triangleq\Set{(v,e)\mid v\in U(t_j), \text{$e$ is incident to $t_j$ in $G_j$}}$.

Recall that a realization of a request is a path in $N$. Given a
realization and weights $w$ over nodes and edges of $N$, we define the
\emph{weight of a realization} $p$ of $r_j$ is defined to be
$w(p)\triangleq \sum_{x\in
  p} d_j\cdot m_p(x)$, where $m_p(x)$ denotes the number of times node or edge $x$
appears in $p$.  The weight of a path $q$ in $\ppn(N,r_j)$ is simply the sum of the
edge weights in $q$.

The following lemma states the main property of the construction of
$\ppn(N,r)$. The proof contains definitions of the functions $\fold$
and $\expand$ that convert
between paths in $N$ and $\ppr(N,r_j)$.

\begin{lemma}
\label{lem-reduction}
Let $N=(V,E,w)$ be a physical weighted network and let $r_j$ be a request.
There is a one-to-one weight preserving correspondence between valid
realizations of $r_j$ in $N$ and  simple paths in
$\ppn(N,r_j)$ that start in a vertex of $S_j$ and end in a vertex of $T_j$.
\end{lemma}
\begin{sketch}
  Define a function $\fold$ to map a path $p'_j$ in the product graph to a
  realization $p=\fold(p'_j)$  by the following local transformation:
  a processing edge $\big((v,y),(v,y')\big)$ is contracted to the node
  $v\in V$; each routing edge $\big((v,e),(v',e)\big)$
  of $p'_j$ is replaced by the edge $(v,v')\in E$. Clearly, $p_j$ is a valid
  realization of $r_j$ under the segmentation that segments $p_j$ at the nodes
  representing both ends of a contracted processing edge.

  Conversely, assume that a valid realization is given, where
  $p_j=(v_0,\ldots, v_k)$ is the path with segmentation
  $p_j=p_j^1\circ p_j^2\circ\cdots\circ p_j^\ell$. We define
  $p'_j=\expand(p_j)$ as follows.  By assumption, each subpath $p^i_j$ is a valid realization of some \pr-edge
  $e_j^i=(x_{i-1},x_i)$, and such that the endpoint of subpath $p_j^i$ is in
  $U_j(x_i)$. To obtain $p'_j=\expand(p_j)$, apply the following mapping. Map each
  edge $(v,v')$, say in the $i^{th}$ subpath of $p_j$, to the routing edge
  $\big((v,e_j^i),(v',e_j^i)\big)\in E'$, and map each endpoint $v$ of subpath
  $i<\ell$ to the processing edge $\big((v,e_j^i),(v,e_j^{i+1})\big)\in E'$. Clearly,
  $p'_j$ is a path in $N'$, and it connects a node in $S_j$ with a node in $T_j$
  because $p_j'$ starts with a node $(s_j',e)$ for some $s_j' \in U_j(s_j), e\in Y_j$ and ends with a
  node $(t_j',e')$ for some $t_j' \in U_j(t_j), e'\in Y_j$.

  It is straightforward to verify that \fold\ and \expand\  preserve weights.
\end{sketch}\\
We note that an edge
or a node of $N$ might be mapped to at most $k$ times by $\fold$,
where $k$ is the length of the longest simple $s$-$t$ path in
the \pr-graph $G_j$.

\subsection{The Oracle}\label{sec:oracle}
We refer to the algorithm which
computes the realization as an \emph{oracle}. The oracle's
description is as follows.

We are given a request $r_j=(G_j,d_j,b_j, U_j)$
and a weighted physical network $N$. We then apply the following procedure to
find a valid realization of $r_j$ in $N$. 
\begin{compactenum}
  \item $N' \gets \ppn(N,r_j)$.
  \item\label{st-sp} Let $P_j$ denote the set of simple paths in $N'$
  that (1)~start
  in a node in $\Set{(v,e)\mid v\in U_j(s_j)}$,
    (2)~end in a node in $\Set{(v,e)\mid v\in U_j(t_j)}$, and (3)~have
    weight at most $b_j$.
  \item  Let $\Gamma_j \gets \{\fold(p') \mid p' \in P_j\}$.
  \item  Return an arbitrary path $p_j\in \Gamma_j$, or ``FAIL'' if
    $\Gamma_j=\emptyset$.
\end{compactenum}
 Step~\ref{st-sp} can be implemented by any shortest-paths
algorithm, e.g., Dijkstra's. Note that the  oracle ignores the demand
$d_j$ (and thus does not
verify that the returned path $p_j$ satisfies the capacity
constraints; feasibility will follow from the weight assignment and
the assumption on the maximal demand).

 
\section{The  Algorithm}
\label{sec-alg}
In this section we first describe our algorithm in Section
\ref{sec:gen}, then analyze it in Section~\ref{sec-analysis}, and finally present
a lower bound to the problem in Section~\ref{sec:proof}.

To solve the problem described in Section \ref{sec-serv-model},  we
employ the resource allocation algorithm
of~\cite{DBLP:conf/icdcn/EvenMSS12,onmcf12} (which
extends~\cite{AAP,BN06}). 
The general idea is as follows. We assign weights to nodes and
edges according to their current load.
For each incoming request, a realization in the network is found
as described in Sec.~\ref{sec:oracle}, and
 submitted to the resource allocation algorithm.
If that algorithm decides to accept the request, our algorithm
algorithm accepts; otherwise the request is put in the
standby mode, and it will be tried again when any accepted
request leaves.

We assume for now that (1) the
allocations of resources to requests are simple paths, and that (2) the demand function is defined only
over
edges and that all demands in a given request are equal.
We lift these restrictions in Sec.~\ref{sec:ext}.

\paragraph{Terminology.}\label{sec:term}
Let $k$ denote an upper bound on the length of a longest simple path in the
\pr-graphs.
Let $\pmax$ denote an upper bound on the length of valid realizations (clearly,
$\pmax < |V|k$).
Let $\bmax$ denote an upper bound on the benefit per time unit offered
by any request.
Define
\begin{equation}
  \label{eq:af}
  \af\eqdf{\log(3\pmax\bmax+1)}
\end{equation}
Note that $\Phi=O(\log n+\log b_{\max}+\log k)$.

A \emph{feasible path} $p$ for request $r_j$ is a path which is a valid realization
of $r_j$ with minimum edge capacity at least $d_j\cdot (3k\Phi)$.  We denote the set
of feasible paths for request $r_j$ by $\Gamma_j$.

We say that a request $r$ is \emph{active} at time $t$ is $t$ has
arrived before time $t$ and has not departed by time $t$.
Given time $t$ in the run of the algorithm and an edge $e$ of $N$,
 $f(e)$ denotes the sum of demands of accepted active  requests that are routed over $e$. 
Recall that the \emph{load} of an edge $e$ is defined by
$
  \load(e) \eqdf \frac{f(e)}{c_e}\:.
$
The \emph{exp-load} of $e$ is defined by
\begin{align} \label{eq:exp-load}
  x_e &\triangleq \frac{1}{\pmax} \cdot \left(2^{\load(e)\cdot\af}-1\right).
\end{align}

\subsection{Algorithm Operation}
\label{sec:gen}

\newcommand{\newproc}[1]{\vspace{7pt}\\{}\hspace{-6mm}\underline{#1}}
\begin{algorithm}
\caption{\small \alg\ - an online algorithm for the SDN
problem.
\ifnum\submit=0
It is assumed that the demand of all requests
is at most $\min_{z\in N} c_z/(3k \cdot\af)$, where $c_z$
is the capacity of an edge/vertex $z$, $k$ is the maximal
$s$-$t$ path length in the \pr-graph of a request, $\pmax$
is the maximal path length in $N$, and $\bmax$ is the
maximal benefit per time unit that can be paid by a
request.
\fi
}\label{alg:alg} \small

\algblock{Begin}{End}
\begin{algorithmic}[1]
	\small

\Statex \hspace{-4mm}\underline{State:}
\begin{compactitem}
	\item  $L$: a set, contains all unserved  active requests (in
          standby mode)
	\item $A$: a set, contains all served active requests.  Each request
	$r_j\in A$ is routed over
	a path $p_j$.
\end{compactitem}
\Statex

\vspace{-2mm} \Statex \hspace{-4mm}\underline{Actions:}
\Statex\hspace{-2mm}Upon arrival of request $r_k$: \Begin
  \State \routealg($r_k$)  \label{line call route 1}
  \If {$r_k\notin A$}   \Comment{request not accepted now}
  \State $L \gets L\cup \{r_k\}$ ; output ``$r_k$: standby''
  \label{line add}
  \EndIf
\End

\Statex \vspace{-2mm} \Statex \hspace{-2mm}Upon departure
of request $r_k$: \Begin \State
\unroutealg($r_k$)\label{line call unroute} \State
\textbf{for all} {$r_j \in L$} \textbf{do}
\routealg($r_j$). \Comment{all
orders allowed} \label{line call route 2} \End \Statex
\vspace{-2mm} \Procedure{\routealg($r_j$)}{}
\State Invoke $\AAP(r_j)$
\If{\AAP\ returned a path $p_j\ne\bot$}
\State $A\gets A\cup \{r_j\}$;
 $L \gets L \setminus \{r_j\}$
 \State output ``$r_j$ accepted, path $p_j$''
\EndIf
\EndProcedure \Statex \vspace{-2mm}
\Procedure{\unroutealg($r_k$)}{} \If {$r_k \in A$} \State
$A\gets A\setminus \{r_k\}$
\For {all $e\in p_k$} \State let $m(e)$ be the multiplicity of $e$  in
$p_k$ \State  $f(e) \leftarrow
f(e)-m(e) \cdot d_k$ \Comment{Free the path
$p_k$}\label{line free} \EndFor \Else \State $L \gets L
\setminus \{r_k\}$
\label{line unstandby 2} \EndIf \EndProcedure \Statex
\vspace{-2mm}
\Procedure{$\AAP(r_k)$ {\rm where} $r_k=(G_k,d_k,b_k)$}{}
\State Assign to each edge $e$ in $N$ weight $x_e$
\Comment{$x_e$ defined in Eq.~\ref{eq:exp-load}}

\State Let $C_k$ be a subset of $\Gamma_k$ whose weight is at most $b_k$
\label{line:oracle1}

\State Pick an arbitrary path $p_k\in C_k$ \Comment{Invocation of the oracle. See
  Sec.~\ref{sec:oracle}} \label{line:pick}

\For {all $e\in p_k $}

\State let $m(e)$ be the multiplicity of $e$  in $p_k$
\State $f(e) \leftarrow
f(e)+m(e) \cdot d_k$
\Comment{update loads}\label{line allocate}

\EndFor

\State \Return $p\in C_k$, or $\bot$ if $C_k=\emptyset$
\EndProcedure
\end{algorithmic}
\end{algorithm}

Pseudo-code for the algorithm, called \alg, is provided in
Algorithm~\ref{alg:alg}. 
The algorithm maintains a set  $L$ of the requests in standby mode:
these are all
active requests currently not served. The set $A$ contains all active
requests currently served.
The path allocated for an active request $r_j$ is denoted
by $p_j$.  When a request arrives, the algorithm tries to
route it by calling \routealg. If it fails, $r_j$ is
inserted into $L$ (line~\ref{line
add}). A departure of an active request is handled by
invoking the \unroutealg\ procedure (line~\ref{line call
unroute}), and then the algorithm tries to serve every
 request in $L$ by invoking \routealg\ (line~\ref{line call
route 2}).  Any order can be used to try the standby requests, thus
allowing for using arbitrary dynamic  priority policies.
The \routealg\ procedure calls \AAP, which is an online procedure for a generalization of the path-packing problem (see below). If \AAP\ allocates a path $p_j$ in $N$, then $r_j$ is accepted.
Otherwise $r_j$ is inserted to the standby list $L$.
Procedure \AAP\ first searches for a path in $N$ which is a realization of the request.  The weight of a path is
defined as the sum of the exp-loads of the edges along it.
If a path whose weight is less than the benefit $b_j$ is
found, then the request is allocated. The task of finding
such a path is done via the oracle described
in Section~\ref{sec-reduction}.
The \unroutealg\ procedure removes a request from $A$
or $L$; if the request was receiving service, the load of its edges is
adjusted (line~\ref{line free}).

\subsection{Analysis}
\label{sec-analysis}
We compare the
performance of \alg with an offline fractional optimal solution, denoted by
$\optpack_f$.
More precisely, we compare the benefit produced by
\alg with the benefit and load of \emph{any} allocation that respects
the capacity constraints. Such allocations may
serve
a request partially and obtain the prorated benefit, and
may also split the flow of one request over multiple paths.
Among these allocations,  $\optpack_f$ denotes the allocation that
achieves the maximal benefit.
Moreover, in
each time step, $\optpack_f$ induces a new multicommodity flow
(independent of the flow of $\optpack_f$ at any other time step).
Implicitly, this means that $\optpack_f$
may also arbitrarily preempt and resume requests, partially or wholly.

Given time step $t$, let $\benefit_t(\alg)$ denote
the benefit to \alg due to step $t$, and
analogously, let $\benefit_t(\optpack_{f})$ denote the benefit gained
by $\optpack_f$ in time step $t$. The competitiveness of \alg is
stated in the following theorem.

\begin{theorem}\label{thm:repaap}
Let $z$ range over nodes and edges of $N$.
If  $\max_j
d_j \leq \min_{z}c_z /(3k\af)$ for each request $r_j$, then
$\benefit_t(\alg) \geq\frac1{3\af}\cdot \benefit_t(\optpack_f)$
  in each time step $t$.
\end{theorem}

The proof of Theorem~\ref{thm:repaap} is based on an analysis of Procedure \raper\
in each time step (which is analogous to an analysis with respect to persistent
requests). This analysis appears in Sec.~\ref{sec:ext}. The proof of
Theorem~\ref{thm:repaap} appears in Sec.~\ref{sec:proof}.

\subsubsection{Online Resource Allocation with Persistent Requests}\label{sec:ext}

We now present another key ingredient in our solution, namely the
\emph{online resource allocation} problem with respect to persistent
requests (and the classical accept/reject service model). The
algorithm to solve it is a generalization of~\cite{AAP}.

By online resource allocation we mean the following setting. Consider
a set $E$ of $m$ \emph{resources}, where each resource $e \in E$ has a
capacity $c_e$. Requests $\{r_j\}_j$ arrive in an online fashion. Each
request $r_j$ specifies a set of possible \emph{allocations} denoted
$\Gamma_j\subseteq 2^E$.  Let $\pmax$ be an upper bound on the number
of resources in every feasible allocation. i.e., $|p|\le \pmax$ for
all $p\in\bigcup_j\Gamma_j$.  We allow a general setting in which the
demand requirement depends on the request, the allocation, and the
resource. Formally, the demand of the $j$th request with respect to
allocation $p$ is a function $d_{j,p}: p \rightarrow \NN$.  Each
request $r_j$ has a benefit $b_j$ that it pays if served.  Let $\bmax$
denote an upper bound on $\max_j b_j$, and $\af \triangleq \log
(1+3\pmax\bmax)$.  Small demands mean that $d_{j,p}(e)/c_e\leq
1/(3\af)$ for every $e\in p$.

\medskip\noindent Resource allocation generalizes allocations in
circuit switching networks and SDNs:
\begin{compactitem}
  \item In the virtual circuits problem, $\Gamma_j$ is
      simply the set of feasible paths from the source to
      the destination.
    \item For an SDN request $r_j = (G_j,d_j,b_j,U_j)$, the set of
      allocations $\Gamma_j$ is simply the set of valid realizations,
      each of which is a subset of network edges and nodes. Since a
      realization $p$ may contain cycle, resources may be used more
      than once by $p$. Therefore, the load on a resource $e\in p$
      incurred by a realization $p$ of request $r_j$ is $d_j\cdot
      m_p(e)$, where $m_p(e)$ denotes the number of times $e$ appears
      in $p$.
\end{compactitem}

\medskip In the general resource allocation problem, $\Gamma_j$ does
not need to have any particular structure, but an algorithm to solve
it must have an
\emph{oracle}. The task of the oracle is as follows.
Assume that every resource $e$ has a \emph{weight} $x_e$. The weight
of an allocation $p\subseteq E$ is simply $x_p \triangleq \sum_{e\in
  p} x_e$.  Given a request $r_j$ with benefit $b_j$, the oracle
returns either an allocation $p\in \Gamma_j$ such that $x_p< b_j$, or returns
``FAIL'' if no such allocation exists.

The online algorithm \raper\ for the resource allocation problem with
persistent requests uses a modified variant of Procedure \AAP\ for each
request. The modifications of \AAP\ are as
follows.  First, the term ``path'' is interpreted as a feasible
allocation.  Second, lines~\ref{line:oracle1}-\ref{line:pick} are
replaced by an invocation of the oracle.  If a feasible
allocation of weight at most $b_j$ is found, then the flow is updated
accordingly (as in Line~\ref{line allocate}).

An online resource allocation algorithm for persistent requests
appears in~\cite{DBLP:conf/icdcn/EvenMSS12,onmcf12} (this algorithm is
extends the path packing algorithm for persistent requests
of~\cite{AAP} using the analysis in~\cite{BN06}). Since Algorithm
\raper\ is a special case of this extension, we obtain the following
theorem.

\begin{theorem}[\cite{DBLP:conf/icdcn/EvenMSS12, onmcf12} based on
  \cite{AAP,BN06}]
\label{thm-aap}\label{thm:AAP}
Let $N$ be a given network and let $\sigma=\Set{r_j}$ be a sequence of
persistent requests. If $d_{j,p}(e) \leq c_{e}/(3 \cdot \af)$, for
every request $r_j$, every allocation $p\in \Gamma_j$, and every
resource $e$, then $\raper(\sigma)\ge\optpack^f(\sigma)/(3 \cdot
\af)$.
\end{theorem}

\ifnum\disc=1 \vspace{-17pt} \fi
\paragraph{Application to SDN requests.}
The oracle for finding light-weight feasible realizations for SDN
requests finds a lightest path in the weighted product network (see
Sec.~\ref{sec:oracle}). Folding such a path may result in a
realization with cycles in the SDN
network. Multiple occurrences of an edge or node $z$ in a
realization $p$ means that the load on $z$ is multiplied by the number
$m_p(z)$ of occurrences of $z$ in $p$. Hence the demand from $z\in
E\cup V$ induced by a realization $p$ for request $r_j$ satisfies
$d_{j,p}(z) = d_j\cdot m_p(z)$. Let $k$ denote an upper bound on the
number of processing stages in realizations (i.e., $k$ equals the
length of a longest simple source-sink path in the \pr-graphs). Then
 $m_p(z)\le k$.  Since  Procedure \AAP\
 considers unfolded paths for its decisions, the requirement that
$\max_j d_j \leq \min_{z}c_z /(3k\af)$ allows us to
apply Theorem~\ref{thm:AAP} to procedure \AAP\ in the context of SDN
persistent requests. 

Regarding persistence, observe that the benefit in a single time step
where each request pays $b_j$ per served time step is identical to the
benefit with respect to persistent requests (where each request pays
$b_j$ if served).

\subsubsection{Proof of Theorem~\ref{thm:repaap}}
\label{sec:proof}
\begin{proof}[Proof of Theorem~\ref{thm:repaap}]
  The proof proceeds by a simulation argument.  Specifically, we
  interpret the execution of \alg\ as a repeated execution of the
  \AAP\ algorithm in each time step with respect to the active
  requests.  For the purpose of the simulation, assume that in the
  \AAP\ algorithm, each request $r_k$ may be accompanied by a
  preferred feasible path $p'_k$. If the weight of $p'_k$ is less than
  $b_k$ (i.e., $p'_k\in C_k$), then the \AAP\ algorithm allocates the
  preferred path $p'_k$ to $r_k$.

  We prove the theorem by induction on $t$.  The base case $t=1$ follows from
  Theorem~\ref{thm-aap}. For the
  induction step, let $A_t$ and $L_t$ denote the
  sets $A$ (served requests) and $L$ (pending requests)
  at time $t$, respectively. The event
  $\sigma_{t+1}$ that occurs in step $t+1$ is either an arrival of a
  new request
  or the departure of an active request.

  Suppose first that $\sigma_t$ is an arrival of a new request
  $r_k$. In this case, we simulate the \AAP\ algorithm by feeding it
  first with the requests in $A_t$ according to the order in which
  they were served (this order may be different than the order in
  which they arrived). Each request $r_j\in A_t$ is accompanied by a
  preferred path $p'_j$, where $p'_j$ is the path that was allocated
  to $r_j$ in time step $t$. The flow along each edge when a request
  is introduced to the \AAP\ algorithm is not greater than the flow
  along the edge just before \alg\ first served the same
  request. Hence, the weight of $p'_j$ during this simulation of step
  $t+1$ does not exceed its weight when it was first allocated to
  $r_j$, and hence $p'_j\in C_j$ also in time step $t+1$. This implies
  that the \AAP\ algorithm accepts all the requests in $A_t$ and
  routes each one along the same path allocated to it in the previous
  time step. Next we feed $r_k$ to the \AAP\ algorithm (without a
  preferred path).  The result of this simulation is identical to the
  execution of \alg\ in time step $t+1$. By Theorem~\ref{thm:AAP}, the
  theorem holds for step $t+1$ in this case.

  We now consider the case that $\sigma_t$ is a departure of an active
  request $r_k$.  In this case we may simulate the \AAP\ algorithm by
  feeding it first with the served requests in $A_t\setminus \{r_k\}$
  according to the order in which they were served, each request
  accompanied with its preferred path. Again, all these requests will
  be served by their preferred paths. After that, the pending requests
  in $L_t$ are input to the \AAP\ algorithm in the same order that
  they are processed by \alg\ in  step $t+1$; each pending request
  that is served by \alg\ in  step $t+1$ is accompanied by a
  preferred path that equals the path that is allocated to it by
  \alg. As the states of \alg\ and the simulated \AAP\ algorithm are
  identical, the \AAP\ algorithm accepts the same pending requests and
  serves them along their preferred paths. By Theorem~\ref{thm:AAP},
  the theorem holds for  step $t+1$ in this case as well.
\end{proof}
Note that the proof easily extends to the case that multiple events occur in each
time step.


\ifnum\disc=1 \vspace{-7pt} \fi


\section{Lower Bound}
\label{sec-lb}
In this section we state a lower bound which implies that the competitive ratio of \alg is asymptotically optimal, up to an additive $O(\log k)$, where $k$ is the maximal length of a source-sink path in any \pr-graph. We note that with current technology, $k$ is typically a small constant.

\begin{theorem}\label{thm:lbaap}
  Every online algorithm in the \accstdby\ service model is $\Omega(\log (n \cdot \bmax))$-competitive.
\end{theorem}
The idea of the proof is to reduce the bad scenarios for the
persistent request model \cite{AAP} to bad scenarios in the \accstdby\
model.  To do that, consider requests arriving one per time unit,
followed by some $T$ time units in which no event occurs, after which
all the requests leave. Observe that since all requests leave
together, there is no
advantage in placing a request at standby: we might as well reject it
immediately.  Hence, the $\Omega(\log n)$ lower bound
of~\cite[Lemma~4.1]{AAP} applies in this case. Using similar
techniques as in~\cite{AAP} it can be showed that every online
algorithm is also $\Omega(\log \bmax)$ competitive.


\bibliographystyle{abbrv}
\bibliography{reroute}

\end{document}